\documentclass[sigconf]{acmart}

\usepackage{enumitem}
\usepackage{array}
\usepackage{algorithm}
\usepackage[noend]{algpseudocode}
\usepackage{xcolor}
\usepackage{tikz}
\usepackage{balance}
\usetikzlibrary{arrows.meta}
\usepackage[nameinlink,noabbrev]{cleveref}

\makeatletter
\renewcommand{\theHALG@line}{\arabic{algorithm}.\arabic{ALG@line}}
\makeatother

\setcopyright{none}
\renewcommand\footnotetextcopyrightpermission[1]{}
\acmConference{}{}{}
\acmYear{}

\definecolor{ink}{HTML}{17324D}
\definecolor{navy}{HTML}{234A73}
\hypersetup{
    colorlinks=true,
    linkcolor=navy,
    citecolor=navy,
    urlcolor=navy
}
\setlist{nosep,leftmargin=1.5em}

\newtheorem{theorem}{Theorem}
\newtheorem{lemma}[theorem]{Lemma}

\newtheorem{claim}[theorem]{Claim}

\crefname{theorem}{theorem}{theorems}
\Crefname{theorem}{Theorem}{Theorems}
\crefname{proposition}{proposition}{propositions}
\Crefname{proposition}{Proposition}{Propositions}
\crefname{lemma}{lemma}{lemmas}
\Crefname{lemma}{Lemma}{Lemmas}
\crefname{corollary}{corollary}{corollaries}
\Crefname{corollary}{Corollary}{Corollaries}
\crefname{claim}{claim}{claims}
\Crefname{claim}{Claim}{Claims}

\title{Non-Clairvoyant Scheduling is Hard Even for Trees
}

\author{Kunal Agrawal}
\authornote{Supported in part by NSF grants 2216971, 2106699, and 2107280.}
\affiliation{\institution{Washington University in St. Louis}
  \city{St. Louis}
  \state{Missouri}
  \country{USA}
}
\email{kunal@wustl.edu}

\author{Owen Druzgal}
\affiliation{\institution{University of Michigan}
  \city{Ann Arbor}
  \state{Michigan}
  \country{USA}
}
\email{odruzgal@umich.edu}

\author{Milind Prabhu}
\authornote{Supported by NSF award CCF-2327011.}
\affiliation{\institution{University of Michigan}
  \city{Ann Arbor}
  \state{Michigan}
  \country{USA}
}
\email{milindpr@umich.edu}

\author{Jinhao Zhao}
\affiliation{\institution{Washington University in St. Louis}
  \city{St. Louis}
  \state{Missouri}
  \country{USA}
}
\email{jinhaoz@wustl.edu}
\newcommand{\OPT}{\textsc{Opt}}
\newcommand{\opt}{\operatorname{opt}}
\newcommand{\ALG}{\mathrm{ALG}}

\newcommand{\FIFO}{\mathrm{FIFO}}
\newcommand{\SFF}{\mathrm{SFF}}

\newcommand{\alg}{\mathrm{ALG}}

\begin{document}
\begin{abstract}
We study online scheduling of parallel jobs on $m$ identical processors.
Jobs arrive over time, and each job $J$ is modeled as a directed acyclic
graph (DAG) whose vertices are subjobs. Each subjob requires one unit of
processing on a single processor, and an edge $u \rightarrow v$ requires $u$ to
complete before $v$ can be processed. The goal is to design an online
algorithm that minimizes the maximum flow time, i.e., the maximum time
between the arrival and completion of any job.

We focus on non-clairvoyant online algorithms, which do not know a job's DAG
when it arrives. Instead, the DAG is revealed online: a subjob becomes known
only when it is ready, that is, after all its predecessors have completed.
First-In-First-Out (FIFO) is
a natural non-clairvoyant algorithm that prioritizes ready subjobs from
earlier-arriving jobs. Agrawal, Moseley, Newman, and Pruhs (SPAA 2024)~\cite{agrawal24} showed that FIFO
has competitive ratio $\Omega(\log m)$ even when every job is an out-tree.
They further showed that FIFO is $O(\log m)$-competitive in several
natural settings and asked whether this bound holds in general. They also
raised the question of whether a different non-clairvoyant algorithm
exists that is~$O(1)$\nobreakdash-competitive.

We answer both questions in the negative by proving a lower bound
of $\Omega(\min\{m,\opt\})$ for every non-clairvoyant online
algorithm.
Here, $\opt$ denotes the maximum flow time of an optimal offline schedule.
In particular, every non-clairvoyant algorithm has competitive ratio
$\Omega(m)$ on some instance with $\opt \ge m$.
The lower bound holds even when every job is an out-forest.

We complement these lower bounds with an asymptotically optimal
non-clairvoyant algorithm. A simple argument shows that FIFO is
$O(m)$-competitive. However, FIFO can have competitive ratio $\Omega(m)$
even when $\opt=O(1)$. To handle instances with small $\opt$, we design a
new Small-Frontier-First algorithm that is $O(\opt)$-competitive.
Combining the two algorithms yields an
$O(\min\{m,\opt\})$-competitive non-clairvoyant algorithm.

\end{abstract}

\maketitle
\pagestyle{plain}
\thispagestyle{plain}

\section{Introduction}

We consider the fundamental problem of scheduling parallel programs which arrive online on a parallel machine with $m$ identical processors.  Each job is modeled as a directed acyclic graph (DAG) where vertices represent tasks (sequential strands of instructions) and edges represent dependences between these tasks.  This is a natural model for a large class of parallel programs; for instance, programs that are constructed using fork-join parallelism (also called dynamic multithreading) which is supported by parallel programming languages such as Cilk dialects~\cite{frigo98}, OpenMP~\cite{openmp13}, and Intel's TBB~\cite{reinders10}.  In this setting, the structure of the DAG is often not known in advance and is revealed online as the job is processed.  In particular, a vertex is \emph{ready} to be executed when all its predecessors have been executed and the scheduler only finds out about the existence of a vertex when it is ready.  Therefore, the scheduler must be non-clairvoyant; it doesn't know anything about the jobs when they arrive; vertices are revealed when they become ready and the scheduler must proceed with this limited information.

We study the setting where these parallel jobs arrive online over time, the scheduler must operate non-clairvoyantly, and the goal is to minimize the \emph{maximum flow time} or \emph{maximum latency} --- that is, the maximum time between the arrival of a job and its completion.  Minimizing maximum flow time is a standard quality of service metric since it optimizes for the worst-case quality of service over all jobs.  The problem of minimizing maximum flow time has been studied for sequential jobs~\cite{ambuhl05} and parallel jobs in several different models of parallelism~\cite{agrawalLLM16spaa,pruhsRS10,moseleyZZ22}.

Intuitively, a scheduler for this metric should ``prioritize'' old jobs and this intuition turns out to be the correct one for most settings. A \emph{First-In-First-Out (FIFO) scheduler} operationalizes this intuition; it schedules all ready vertices (tasks) of the oldest job first and then the vertices of the second oldest job and so on until all $m$ processors are busy or it runs out of ready vertices. For sequential jobs (each job is a chain of vertices) on a single machine, FIFO is optimal~\cite{bender98}, while for sequential jobs on $m$ identical processors, FIFO is $(3-2/m)$-competitive~\cite{bender98,mastrolilli03}.  Similarly, for fully parallelizable jobs, it is also optimal on $m$ identical processors.

  The problem of minimizing maximum flow for parallel DAG jobs was first studied in~\cite{agrawalLLM16spaa}, where the authors show that FIFO is $O(1)$-competitive given $(1+\epsilon)$ times the speed of the optimal algorithm.  The question of how to optimize maximum flow on DAGs without speed augmentation was studied in~\cite{agrawal24} and the authors provide some partial answers.  First, in a somewhat counterintuitive result, they show that FIFO is $\Omega(\log m)$-competitive even on jobs that can be represented as out-trees, which form a simpler class than general DAGs.
  In the same paper, they provide an $O(1)$-competitive algorithm for out-trees if the entire structure of the tree is revealed to the scheduler when the job arrives.  They also provide a somewhat positive result for FIFO, namely, for a restricted setting of batched jobs, FIFO is $O(\log(m\opt))$-competitive, where $\opt$ is the maximum flow time provided by the optimal algorithm.
  However, the best achievable competitive ratio for non-clairvoyant scheduling of DAGs (or even trees) without additional speed remained completely open.

In this paper, we resolve this important open question, mostly in the negative direction. Our main result shows that no non-clairvoyant scheduler can achieve a competitive ratio asymptotically better than $\min\{m,\opt\}$, even for trees. The main theorem of the paper is the following:

\begin{theorem}\label{thm:non-clairvoyant-lb}
    Every deterministic online non-clairvoyant scheduler and every randomized online non-clairvoyant scheduler against an oblivious adversary have a competitive ratio of $\Omega(\min\{m, \opt\})$ even for trees.\footnote{To be more precise, the construction uses out-forests; but the main point is that each vertex has at most 1 parent.}
\end{theorem}

This result is quite surprising.  Maximum flow is generally an easier objective than other quality of service objectives such as average flow time and stretch.  For instance, for sequential jobs on a single machine, minimizing average flow time non-clairvoyantly (without knowing the processing times of each job when it arrives) has competitive ratio lower bounds of $\Omega(n^{1/3})$ for deterministic non-clairvoyant schedulers and $\Omega(\log n)$ for randomized schedulers against an oblivious adversary, where $n$ is the number of jobs~\cite{motwani94}.\footnote{SRPT is optimal if the processing times are known on arrival~\cite{schrage68}.}  Moreover, on $m\ge2$ identical machines, even when processing times are revealed on release, every deterministic or randomized online scheduler has competitive ratio lower bounds of $\Omega(\log(n/m))$ and $\Omega(\log P)$, where $P$ is the ratio of the largest to the smallest processing time~\cite{leonardiRaz07}.  Thus, without speed augmentation, no constant competitive ratio is possible for average flow time even for sequential jobs.  In contrast, FIFO is constant-competitive without speed for maximum flow~\cite{bender98,mastrolilli03}.  Therefore, \emph{a priori} one might imagine that some non-clairvoyant algorithm could work for maximum flow without speed, at least for jobs with a simple dependence structure like trees or forests.

In addition, in~\cite{agrawal24}, the authors show an $O(\log(m\opt))$ upper bound for FIFO for batched arrivals of jobs.  Say that $\opt$ is the maximum flow provided by the optimal scheduler.  If jobs only arrive at intervals of $\opt$---that is, sets of jobs arrive at time $0, \opt, 2\opt, ...$, with no arrivals in between---then FIFO provides a competitive ratio of $O(\log(m\opt))$ for general DAG jobs (not just trees).  We expected that this result could generalize to nonbatched arrivals since one can delay the arrivals of jobs to the next multiple of $\opt$ by increasing the maximum flow time of the optimal algorithm only by a factor of $2$.  The authors of the paper indicate that their analysis techniques do not generalize to the nonbatched setting and one would need additional insights to generalize the result.  Our lower bound shows that, in fact, these results cannot generalize.  It is counterintuitive, however, that batched arrivals give an exponential advantage to FIFO in terms of the achievable competitive ratio.

We also prove matching upper bounds.  We design a new non-clairvoyant scheduler that is a combination of FIFO (which by itself provides a competitive ratio of $O(m)$) and an algorithm called \emph{Small-Frontier-First} which tries to execute bottleneck vertices in DAGs in order to reveal parallelism quickly (and by itself achieves a competitive ratio of $O(\opt)$) and show the following result.  

\begin{theorem}
\label{thm:combined}
There is a deterministic non-clairvoyant algorithm for online DAG scheduling with
competitive ratio $O(\min\{m,\opt\})$. 
\end{theorem}

FIFO alone does not attain the preceding guarantee. Indeed, although FIFO is
$O(m)$-competitive, it is not $O(\opt)$-competitive. In particular, there
are instances with $\opt=O(1)$ on which FIFO has competitive ratio
$\Omega(m)$.

\begin{theorem}
\label{thm:fifo-small-opt-lb}
    FIFO provides a competitive ratio of $\Omega(m)$ even if $\opt$ (the maximum flow time provided by the optimal scheduler) is a constant.
\end{theorem}

Finally, we return to the batched-arrival setting discussed above, in which
jobs arrive only at times $0,\opt,2\opt,\ldots$, and determine its optimal
competitive ratio. Using essentially the argument of
Agrawal, Moseley, Newman, and Pruhs~\cite{agrawal24}, but with a more careful
analysis, we sharpen their $O(\log(m\opt))$ upper bound to $O(\log\opt)$ and
prove a matching lower bound.

\begin{theorem}
\label{thm:batched-optimal}
The optimal competitive ratio for non-clairvoyant online scheduling in the
batched-arrival setting is $\Theta(\min\{m,\log\opt\})$. FIFO achieves the
upper bound.
\end{theorem}

Together with previous work, this reveals a sharp threshold in the arrival
spacing. Let $g$ denote the spacing between consecutive batches.
{\centering
\small
\setlength{\tabcolsep}{4pt}
\begin{tabular}{@{}|
    >{\centering\arraybackslash}m{0.26\linewidth}|
    >{\centering\arraybackslash}m{0.35\linewidth}|
    >{\centering\arraybackslash}m{\dimexpr0.29\linewidth-4\arrayrulewidth\relax}|
    @{}}
\hline
\textbf{Arrival Spacing} & \textbf{Competitive Ratio} & \textbf{Source} \\
\hline
$g\le(1-\epsilon)\opt$ for any fixed $\epsilon>0$
& $\Theta(\min\{m,\opt\})$
& This paper \\
\hline
$g=\opt$
& $\Theta(\min\{m,\log\opt\})$
& This paper \\
\hline
$g\ge(1+\epsilon)\opt$ for any fixed $\epsilon>0$
& $\Theta(1)$
& FIFO with speed augmentation~\cite{agrawalLLM16spaa} \\
\hline
\end{tabular}
\par}

Thus, the table reveals a somewhat surprising phenomenon: as the arrival
spacing crosses $\opt$, the best achievable competitive ratio changes from
linear in $\opt$ below the threshold, to logarithmic at the threshold, and
to constant above it.

\subsection{Technical Overview}
\paragraph{General Lower Bound.}
The high-level idea is to construct an instance $I$ where the optimal
offline schedule keeps processors busy by overlapping work across jobs,
while a non-clairvoyant algorithm is repeatedly forced to leave processors
idle.

For an integer $L$, the adversary releases a job at each time $iL$. Each
job is a layered out-forest with $m\cdot L$ subjobs. Each layer is either
\emph{wide}, with $2m$ subjobs, or \emph{narrow}, with one subjob, except
possibly for a smaller final layer. Importantly, each layer has a \emph{key}
subjob that is the predecessor of every subjob in the next layer; all other
subjobs are leaves.
The instance is designed to have $\opt(I)=\Theta(L)$: a clairvoyant offline scheduler knows each key subjob in advance; therefore it can execute it early to reveal parallelism quickly and keep processors busy.

A non-clairvoyant algorithm faces two difficulties:
\begin{itemize}
    \item The algorithm cannot identify the key. The adversary chooses it
    among the last subjobs processed, so the algorithm finishes a layer
    before learning the next.
    \item The algorithm cannot see whether future layers are wide or narrow.
    The adversary chooses each layer's type when the layer is initially revealed.
\end{itemize}

The adversary can force all alive jobs to encounter narrow layers
simultaneously, causing an idle processor slot. These losses accumulate
into backlog until some job has large flow time, giving the
$\Omega(\min\{m,\opt\})$ lower bound.

\paragraph{Lower Bound for Batched Arrivals.}
In the batched setting, jobs are released at times
$0,\opt,2\opt,\ldots$. To construct a lower-bound instance with
$\opt(I)=L$, we release one job every $L$ timesteps and must give OPT a way
to finish each job before the next arrives. Since the total work of each job is $mL$, OPT must keep all processors busy at all times.  Therefore, while one processor advances
through the keys, the other $m-1$ processors must always have leaves to
process. This forces us to maintain a \emph{prefix property}: every prefix
of a job contains at least as many wide layers as narrow layers. The
adversary can still choose the layer order adaptively and force idle slots,
but it must respect this additional constraint. The resulting geometric
argument gives an $\Omega(\min\{m,\log\opt\})$ lower bound. A complementary
record-backlog argument shows that FIFO is $O(\log\opt)$-competitive: at a
record-high backlog, the number of unfinished batches is logarithmic because
their sequential progress grows geometrically with age. Together with FIFO's
$O(m)$ bound, this gives the optimal $\Theta(\min\{m,\log\opt\})$ ratio.

\paragraph{FIFO: A Simple $O(m)$ Upper Bound.}
FIFO's guarantee follows from a win-win between processing work and making
sequential progress. Whenever fewer than $m$ processors are busy, fewer than
$m$ jobs are alive. Each job contains at most $m\cdot\opt$ work, so the
backlog is bounded. During fully busy periods, FIFO processes work at least
as fast as the optimal schedule; together, these observations bound the
backlog by $O(m^2\cdot\opt)$. Now fix a job. Every timestep either uses all
processors on this job and jobs released no later, or processes every ready
subjob of this job, shortening its longest remaining path. The former steps
are bounded by the backlog divided by $m$, and the latter by $\opt$. Thus
FIFO has maximum flow time $O(m\cdot\opt)$ and competitive ratio $O(m)$.

\paragraph{Why FIFO Fails When OPT Is Small.}
A natural question is whether FIFO is constant-competitive when
$\opt=O(1)$. Surprisingly, it can still have competitive ratio $\Omega(m)$,
even when each job consists of just one or two sequential subjobs followed
by parallel leaves. FIFO prioritizes older leaves over younger sequential
work. The adversary arranges that, when the leaves are exhausted, the
remaining jobs' parallel work is not yet ready. With few jobs alive, this
forces idle processors. Repeating these losses produces large waiting
times, even though every job has constant depth. The optimal schedule avoids
these losses by processing younger jobs' sequential work alongside older
jobs' leaves, finishing every job in constant time.

\paragraph{Small-Frontier-First: Unlock Parallelism Early.}
FIFO's weakness suggests advancing sequential work before older parallel
work has run out. Our Small-Frontier-First algorithm groups arrivals from
$2\opt$-length windows into batches and works in rounds. In each round, it
processes the current ready
frontier of every \emph{small} batch, meaning one with at most $m$ ready
subjobs, and fills spare capacity with work from the oldest \emph{large}
batch. Each batch is small in at most $\opt$ rounds, since every such round
shortens its remaining path, and at most one large batch becomes small per
round. Hence round lengths and the number of small batches are $O(\opt)$.
A round with a large batch is full; otherwise only these small batches
remain, so the remaining work is $O(m\cdot\opt^2)$. Younger small batches
cause only $O(\opt^2)$ interruptions, giving an $O(\opt)$ ratio. Doubling
removes the assumption that $\opt$ is known and, together with FIFO, yields
the matching $O(\min\{m,\opt\})$ ratio.

\section{Preliminaries}
\label{sec:preliminaries}

We now formally define the main problem studied in this paper: online DAG scheduling on multiple processors with the objective of minimizing maximum flow time.

\paragraph{Jobs and Instances.}
In the DAG Scheduling problem there are $m$ identical unit-speed processors. There are discrete timesteps $0,1,2,\ldots$. An instance $I$ of DAG Scheduling  consists of jobs
$J_1,\ldots,J_n$ together with a release time 
$r_i\in\mathbb{Z}_{\geq 0}$ for each $J_i$. Each job $J_i$ is a directed acyclic graph whose vertices are unit-processing-time \emph{subjobs}. For subjobs $u,v \in J_i$, a directed edge $u\to v$
encodes the constraint that subjob $v$ can be processed only after $u$ is processed. 

\paragraph{Constraints on the Schedule.}
A valid schedule $\mathcal S$ of the instance $I$ assigns each subjob of $J_1, \ldots, J_n$ a timestep in $\mathbb Z_{\geq 0}$ subject to the following constraints:
\begin{enumerate}
    \item \textbf{(Release Time Constraints) }For every subjob $u\in J_i$  we have 
    $\mathcal{S}(u)\ge r_i$. 
    \item \textbf{(Precedence Constraints) } For every edge $u\to v$ in $J_i$, we have $\mathcal S(u)<\mathcal S(v)$.
    \item \textbf{(Processor Constraints) }For every $t\in \mathbb Z_{\geq 0}$, the number of subjobs $u$ (across various jobs) such that $\mathcal{S}(u) = t$ is at most $m$. 
\end{enumerate}

\Cref{fig:two-jobs-schedules} shows two valid schedules for two jobs released at the same time on three processors.

A subjob $u\in J_i$ is \emph{ready} at time $t$ if $r_i\le t$, $u$ has not yet been processed, and every predecessor of $u$ has been processed before time $t$. A job $J_i$ is \emph{alive} at time $t$ if $r_i\le t$ and it has at least one unprocessed subjob. The \emph{age} of an alive job $J_i$ at time $t$ is $t-r_i$.

\paragraph{Maximum Flow Objective.}
Let \(\mathcal S\) be a valid schedule of jobs \(J_1,\ldots,J_n\). For a job \(J_i\), let \(t_i\) be the last timestep in which some subjob of \(J_i\) is processed. We define the completion time of \(J_i\) under \(\mathcal S\) as \(
C_i(\mathcal S):=t_i+1.\) \footnote{Intuitively, we set the completion time to $t_i+1$ because processing during timestep $t_i$ occupies the interval $[t_i,t_i+1)$ and therefore finishes only at time $t_i+1$.}
The flow time of \(J_i\) is
\[
F_i(\mathcal S):=C_i(\mathcal S)-r_i,
\]
and the maximum flow time of \(\mathcal S\) is
\(
F_{\max}(\mathcal S):=\max_i F_i(\mathcal S).
\)
Our objective is to find a valid schedule minimizing $F_{\max}$.

\paragraph{Online Setting and Competitive Ratio.}
In the online setting, the scheduling algorithm does not know future job arrivals. At each timestep $t$, it irrevocably selects at most $m$ ready subjobs to process.

For an instance $I$, let $\OPT(I)$ denote an optimal offline schedule and define
\[
\opt(I):=F_{\max}(\OPT(I)).
\]
A deterministic online algorithm $\ALG$ is $c$-competitive if, for every instance $I$,
\[
F_{\max}(\ALG(I))\le c\cdot \opt(I).
\]
A randomized online algorithm is $c$-competitive if, for every fixed instance $I$, we have 
\(
\mathbb E\!\left[F_{\max}(\ALG(I))\right]\le c\cdot \opt(I),
\)
where the expectation is over the internal randomness of $\ALG$.

\paragraph{Non-Clairvoyant Online Algorithms.}
We are particularly interested in \emph{non-clairvoyant} online algorithms. Such an algorithm has the property that the set of subjobs that it schedules at time $t$ is a function only of:
\begin{enumerate}
\item the release times of jobs released by time $t$; and
\item the subjobs that have become ready by time $t$ and precedence relations among them. 
\end{enumerate}
In particular, the decision at time $t$ cannot depend on subjobs that have not yet become ready or on precedence relations involving such subjobs.

First-In-First-Out (FIFO) is the canonical non-clairvoyant algorithm. At each
timestep, FIFO processes $m$ ready subjobs, prioritizing those belonging to
earlier-released jobs and breaking ties arbitrarily. If fewer than $m$
subjobs are ready, it processes all of them.

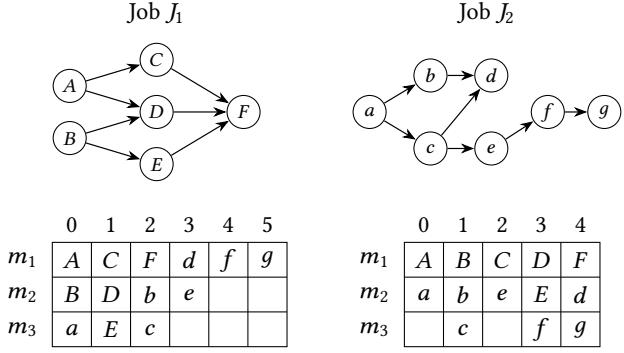
\begin{figure}[t]
\centering
\begin{tikzpicture}[
    x=1cm,y=1.2cm,scale=1.15,transform shape,>=Stealth,
    vertex/.style={draw,circle,minimum size=3.8mm,inner sep=0pt,font=\scriptsize},
    edge/.style={->,line width=0.45pt},
    font=\small
]
\begin{scope}[shift={(0.2,0)}]
    \node at (1,0.95) {Job $J_1$};
    \node[vertex] (A) at (0,0.25) {$A$};
    \node[vertex] (B) at (0,-0.25) {$B$};
    \node[vertex] (C) at (1,0.5) {$C$};
    \node[vertex] (D) at (1,0) {$D$};
    \node[vertex] (E) at (1,-0.5) {$E$};
    \node[vertex] (F) at (2,0) {$F$};
    \foreach \u/\v in {A/C,A/D,B/D,B/E,C/F,D/F,E/F}
        \draw[edge] (\u) -- (\v);
\end{scope}
\begin{scope}[shift={(3.65,0)}]
    \node at (1.35,0.95) {Job $J_2$};
    \node[vertex] (a) at (0,0) {$a$};
    \node[vertex] (b) at (0.7,0.35) {$b$};
    \node[vertex] (c) at (0.7,-0.35) {$c$};
    \node[vertex] (d) at (1.4,0.35) {$d$};
    \node[vertex] (e) at (1.4,-0.35) {$e$};
    \node[vertex] (f) at (2.05,0) {$f$};
    \node[vertex] (g) at (2.7,0) {$g$};
    \foreach \u/\v in {a/b,a/c,b/d,c/d,c/e,e/f,f/g}
        \draw[edge] (\u) -- (\v);
\end{scope}

\begin{scope}[shift={(0,-1.25)},x=0.45cm,y=0.40cm]
    \foreach \t in {0,...,5}
        \node at (\t+0.5,0.5) {$\t$};
    \draw[line width=0.4pt] (0,0) rectangle (6,-3);
    \foreach \t in {1,...,5} \draw[line width=0.4pt] (\t,0) -- (\t,-3);
    \foreach \p in {1,2} \draw[line width=0.4pt] (0,-\p) -- (6,-\p);
    \foreach \p in {1,2,3}
        \node[anchor=east] at (-0.12,-\p+0.5) {$m_{\p}$};
    \foreach \t/\p/\v in {
        0/1/A,0/2/B,0/3/a,
        1/1/C,1/2/D,1/3/E,
        2/1/F,2/2/b,2/3/c,
        3/1/d,3/2/e,4/1/f,5/1/g}
        \node at (\t+0.5,-\p+0.5) {$\v$};
\end{scope}
\begin{scope}[shift={(4.05,-1.25)},x=0.45cm,y=0.40cm]
    \foreach \t in {0,...,4}
        \node at (\t+0.5,0.5) {$\t$};
    \draw[line width=0.4pt] (0,0) rectangle (5,-3);
    \foreach \t in {1,...,4} \draw[line width=0.4pt] (\t,0) -- (\t,-3);
    \foreach \p in {1,2} \draw[line width=0.4pt] (0,-\p) -- (5,-\p);
    \foreach \p in {1,2,3}
        \node[anchor=east] at (-0.12,-\p+0.5) {$m_{\p}$};
    \foreach \t/\p/\v in {
        0/1/A,0/2/a,
        1/1/B,1/2/b,1/3/c,
        2/1/C,2/2/e,
        3/1/D,3/2/E,3/3/f,
        4/1/F,4/2/d,4/3/g}
        \node at (\t+0.5,-\p+0.5) {$\v$};
\end{scope}
\end{tikzpicture}
\caption{Two DAG jobs released at time $0$, with two valid schedules on three processors. The maximum flow time is $6$ in the left schedule and $5$ in the right.}
\Description{Two job DAGs and two valid schedules on three processors. The left schedule completes the jobs by time 6, while the right schedule completes them by time 5.}
\label{fig:two-jobs-schedules}
\end{figure}

\paragraph{Residual DAGs and Basic Properties of $\opt$}
We introduce some additional notation that will be useful in our analysis. For a schedule $\mathcal S$, the \emph{residual DAG} of a job $J_i$ at time $t$ is the subgraph of $J_i$ induced by the subjobs that have not been processed before timestep $t$. We let $w_i^{\mathcal S}(t)$ denote the number of subjobs in this residual DAG, called the \emph{remaining work} of $J_i$ at time $t$, and let $\ell_i^{\mathcal S}(t)$ denote the number of subjobs on a longest directed path in the residual DAG, called the \emph{longest remaining path}. When the schedule is clear from context, we omit the superscript $\mathcal S$. 

When the instance is clear from context, we write $\opt:=\opt(I)$. We will repeatedly use the following simple consequences of the definition of $\opt$.

\begin{proposition}
\label[proposition]{prop:opt-bounds}
For every instance with optimal maximum flow time $\opt$:
\begin{enumerate}
\item every directed path in a job contains at most $\opt$ subjobs;
\item every job contains at most $m\opt$ subjobs; more generally, for integers $0\le a\le b$, the jobs released during timesteps $a,a+1,\ldots,b$ contain at most
\(    m(b-a+\opt)
   \)
subjobs in total.
\end{enumerate}
\end{proposition}

\begin{proof}
The subjobs on a directed path must be processed sequentially. Since every job completes within $\opt$ timesteps of its release under $\OPT$, every directed path contains at most $\opt$ subjobs.

Similarly, a job can receive at most $m\opt$ units of processing during the $\opt$ timesteps between its release and completion, so every job contains at most $m\opt$ subjobs.

More generally, every job released during timesteps $a,a+1,\ldots,b$ must be completed by time $b+\opt$ under $\OPT$. Hence all such work is processed during timesteps $a,a+1,\ldots,b+\opt-1$, which contain exactly $
m(b-a+\opt)
$
processor slots.
\end{proof}

\section{Lower Bound for Non-Clairvoyant Algorithms}
\label{sec:lower-bound}
We now prove \Cref{thm:non-clairvoyant-lb} using a hard instance consisting of out-forest jobs. The key weakness of non-clairvoyance that we exploit is that, among many ready subjobs, the algorithm cannot distinguish leaves from subjobs with many children, whose completion can unlock substantial future parallelism. In the worst case, it may therefore postpone exactly those subjobs. Our construction exploits this uncertainty: the offline optimum knows the full DAGs and can keep the processors essentially fully utilized, while any non-clairvoyant algorithm is periodically forced to incur idle processor-time because it delays the subjobs that would have exposed enough new work to keep all processors busy.

\paragraph{Lower Bound Instance.}
 Given any deterministic non-clairvoyant algorithm  $\alg$, we construct an instance on which the competitive ratio of the algorithm is $\Omega(c)$. We later extend this lower bound to randomized algorithms. 

Fix a sufficiently large integer $L$ divisible by $16$, and let \[c:=\min\{\lfloor m/2 \rfloor , L/16\}, \quad \text{ and } \quad  n:=c^2mL.\]

For $i=0,\ldots,n-1$, job $J_{i+1}$ is released at time $iL$. Each job is an out-forest containing exactly $mL$ subjobs, partitioned into $\Theta(L)$ layers. Each layer, except possibly the last, is either \emph{wide}, containing $2m$ subjobs, or \emph{narrow}, containing a single subjob. The exact sequence of wide and narrow layers is not fixed in advance; instead, the adversary determines adaptively whether each newly exposed layer is wide or narrow. Every layer has one distinguished subjob, called the \emph{key subjob}. The key subjob is the unique predecessor of every subjob in the next layer, while all other subjobs in the layer are leaves. See \Cref{fig:layered-job} for an example of such a job. 

\paragraph{Choosing the Key Subjob.}
When a job arrives, $\ALG$ sees only its first layer. Each subsequent layer becomes ready when the key subjob of the current layer is processed. In every wide layer, the adversary chooses as key one of the last subjobs processed by $\ALG$. The adversary can always do this because $\ALG$ is non-clairvoyant and cannot tell which subjob in the layer is the key. Thus, for every job, $\ALG$ finishes each layer before the next becomes ready.
\paragraph{Choosing the Layer Sizes.}
We now describe how the adversary adaptively fixes the layer sizes in order to force $\ALG$ to leave processor slots idle at regular intervals. We call timesteps $iL,\ldots,(i+1)L-1$ \emph{epoch $i$}. Epochs $0,c,2c,\ldots$ are called \emph{forcing epochs}, and the first $L/2$ timesteps of each forcing epoch form its \emph{forcing interval}.

The first layer of every job is wide. Whenever a subsequent layer becomes ready, the adversary makes it narrow if the current time lies in the first forcing interval during which that job is alive. At all other times, it makes the layer wide if at least $2m$ subjobs remain; if fewer than $2m$ remain, the adversary places all remaining subjobs in one final layer.

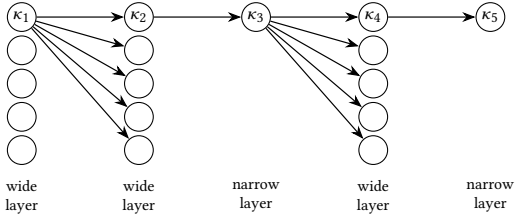
\begin{figure}[t]
\centering
\begin{tikzpicture}[
    >=Stealth,
    vertex/.style={draw,circle,minimum size=3.8mm,inner sep=0pt,font=\scriptsize},
    key/.style={vertex},
    leaf/.style={vertex},
    dependency/.style={->,line width=0.45pt},
    spine/.style={dependency},
    font=\small
]
    \node[key]  (k1) at (0,0.88) {$\kappa_1$};
    \node[leaf] (a1) at (0,0.44) {};
    \node[leaf] (b1) at (0,0) {};
    \node[leaf] (c1) at (0,-0.44) {};
    \node[leaf] (d1) at (0,-0.88) {};

    \node[key]  (k2) at (1.55,0.88) {$\kappa_2$};
    \node[leaf] (a2) at (1.55,0.44) {};
    \node[leaf] (b2) at (1.55,0) {};
    \node[leaf] (c2) at (1.55,-0.44) {};
    \node[leaf] (d2) at (1.55,-0.88) {};

    \node[key]  (k3) at (3.10,0.88) {$\kappa_3$};

    \node[key]  (k4) at (4.65,0.88) {$\kappa_4$};
    \node[leaf] (a4) at (4.65,0.44) {};
    \node[leaf] (b4) at (4.65,0) {};
    \node[leaf] (c4) at (4.65,-0.44) {};
    \node[leaf] (d4) at (4.65,-0.88) {};

    \node[key]  (k5) at (6.20,0.88) {$\kappa_5$};

    \draw[dependency] (k1) -- (a2);
    \draw[dependency] (k1) -- (b2);
    \draw[dependency] (k1) -- (c2);
    \draw[dependency] (k1) -- (d2);
    \draw[spine] (k1) -- (k2);
    \draw[spine] (k2) -- (k3);
    \draw[spine] (k3) -- (k4);
    \draw[dependency] (k3) -- (a4);
    \draw[dependency] (k3) -- (b4);
    \draw[dependency] (k3) -- (c4);
    \draw[dependency] (k3) -- (d4);
    \draw[spine] (k4) -- (k5);

    \node[font=\scriptsize, align=center] at (0,-1.50)
        {wide\\layer};
    \node[font=\scriptsize, align=center] at (1.55,-1.50)
        {wide\\layer};
    \node[font=\scriptsize, align=center] at (3.10,-1.50)
        {narrow\\layer};
    \node[font=\scriptsize, align=center] at (4.65,-1.50)
        {wide\\layer};
    \node[font=\scriptsize, align=center] at (6.20,-1.50)
        {narrow\\layer};
\end{tikzpicture}
\caption{A layered job with wide and narrow layers. The labeled key subjob
at the top of each layer is the unique predecessor of every subjob in the
next layer. All other subjobs are leaves.}
\label{fig:layered-job}
\Description{Five consecutive layers of a job DAG: wide, wide, narrow,
wide, and narrow. Each layer has one labeled key vertex at the top, with
edges to every vertex in the next layer. All other vertices are leaves.}
\end{figure}

\begin{lemma}
\label[lemma]{lem:lower-bound-opt}
The instance $I$ produced above satisfies
$L\le \opt(I)\le 2L$.
\end{lemma}

\begin{proof}
Each job contains $mL$ subjobs, and therefore has flow time at least $L$ in any schedule.

We next give a schedule with maximum flow time at most $2L$. First observe that every job has at most $L$ layers. During the first forcing interval in which a job is alive, at most $L/2$ narrow layers can be created, since the interval has length $L/2$ and at most one new layer can become ready in each timestep. Every other layer contains $2m$ subjobs, except possibly for one final smaller layer, so there are at most $L/2$ such layers. Hence every job has at most $L$ layers in total.

Consider the following offline schedule. During epoch $i$, reserve one processor for $J_{i+1}$. This processor repeatedly processes the key subjob of the current layer until all layers have become ready, and then processes arbitrary ready subjobs of $J_{i+1}$. Since $J_{i+1}$ has at most $L$ layers, by the end of epoch $i$ all of its subjobs are ready, and the reserved processor has processed exactly $L$ subjobs of $J_{i+1}$.

During epoch $i+1$, the other $m-1$ processors process the remaining $(m-1)L$ subjobs of $J_{i+1}$, while the reserved processor begins processing $J_{i+2}$. Thus every job completes by time $(i+2)L$, at most $2L$ timesteps after being released at time $iL$. Therefore $\opt(I)\le 2L$.
\end{proof}
The next lemma shows that, as long as at most $c$ jobs are alive, each forcing interval causes $\ALG$ to waste some processor capacity.
\begin{lemma}
\label[lemma]{lem:forcing-epoch}
Let $i<n$ be a forcing epoch. If every alive job has age less than $cL$ at the beginning of epoch $i$, then the forcing interval of epoch $i$ contains an idle processor-time slot.
\end{lemma}

\begin{proof}
Immediately after $J_{i+1}$ is released, let $q$ be the number of alive jobs.
Since each has age less than $cL$, we have $q\le c$, and this is the first
forcing interval of every such job. Each job initially has at most one ready layer, of size at most $2m$, and every later layer that becomes ready during the interval is narrow. Thus $\ALG$ processes at most $2m+L/2$ subjobs from each job, and hence

$$
q\cdot\left(2m+\frac L2\right)
\le \frac L{16}\cdot 2m+\frac m2\cdot\frac L2
=\frac{3mL}{8}<\frac{mL}{2}.
$$

Here $c\le L/16$ bounds the first term, and $c\le m/2$ bounds the second.
The forcing interval contains $mL/2$ processor-time slots and no further job arrivals. Therefore at least one processor-time slot is idle.
\end{proof}

\begin{proof}[Deterministic case]
Let $B_i$ be the total unfinished work immediately before time $iL$, so $B_0=0$. Each epoch releases $mL$ work and contains $mL$ processor-time slots, so $B_{i+1}\ge B_i$. Suppose, for contradiction, that every job has flow time less than $cL$. Then \Cref{lem:forcing-epoch} applies to every forcing epoch, and at least one processor-time slot is idle in each such epoch. Hence $B_{i+1}\ge B_i+1$ for every forcing epoch $i$.

The forcing epochs among the first $n$ epochs are $0,c,\ldots,n-c$, so there are exactly $n/c=cmL$ of them. It follows that $B_n\ge cmL$. On the other hand, at time $nL$, every job released by time $(n-c)L$ has age at least $cL$ and hence must be complete. Thus only the last $c-1$ jobs can remain unfinished, giving $B_n\le(c-1)mL$, a contradiction.

Therefore some job has flow time at least $cL$. By \Cref{lem:lower-bound-opt}, the instance has optimal maximum flow time $\Theta(L)$. Hence $\ALG$ has competitive ratio $\Omega(c)=\Omega(\min\{m,\opt(I)\})$, proving the theorem.
\end{proof}

We now extend the lower bound to randomized non-clairvoyant algorithms via Yao's principle.

\begin{proof}[Randomized case]
By Yao's principle, it suffices to construct a distribution over fixed instances under which every deterministic non-clairvoyant algorithm has expected competitive ratio $\Omega(c)$.

\paragraph{The Hard Distribution.}
Let $\mathcal I$ be the finite family of fixed instances that can arise from the adaptive construction above as the behavior of the deterministic algorithm varies, and let $N:=|\mathcal I|$. Sample $N$ instances independently and uniformly from $\mathcal I$. We place the sampled instances one after another, shifting their release times so that the first job of each sample is released $2cL$ timesteps after the last job of the preceding sample.

Every instance in the support of this distribution has optimal maximum flow time at most $2L$. Indeed, by \Cref{lem:lower-bound-opt}, each sampled instance can be scheduled with maximum flow time at most $2L$, and the $2cL$ gap between consecutive samples prevents these schedules from overlapping.

Fix a deterministic non-clairvoyant algorithm $\ALG$. Consider the $i$th sampled instance, conditioned on all jobs from the first $i-1$ samples having flow time less than $cL$. By the separation between samples, all earlier jobs are complete before the $i$th sample begins. Conditioned on the history so far, the continuation of $\ALG$ is again a deterministic non-clairvoyant algorithm. By the deterministic lower bound, there is therefore at least one instance in $\mathcal I$ that forces flow time at least $cL$. Since the $i$th sample is uniform over the $N$ instances in $\mathcal I$, this occurs with conditional probability at least $1/N$.

Thus the probability that none of the $N$ samples forces flow time at least $cL$ is at most $(1-1/N)^N\le e^{-1}$. Hence, with $\Omega(1)$ probability, $\ALG$ incurs maximum flow time at least $cL$. Since every instance in the support has $\opt=O(L)$, every deterministic non-clairvoyant algorithm has expected competitive ratio $\Omega(c)$. By Yao's principle, the same lower bound holds for randomized non-clairvoyant algorithms.
\end{proof}

\section{An Optimal Non-Clairvoyant Algorithm}
\label{sec:upper-bounds}

We prove matching upper bounds using two ingredients. FIFO has
competitive ratio $O(m)$, while a round-based algorithm has competitive ratio
$O(\opt)$. We first analyze the latter algorithm assuming that $\opt$ is known
and then remove this assumption by doubling.

\subsection{FIFO is \texorpdfstring{$O(m)$}{O(m)}-Competitive}
\label{sec:fifo-upper-bound}

We now show that FIFO has $O(m)$-competitive ratio.

\begin{theorem}[FIFO upper bound]
\label{thm:fifo-upper-bound}
The competitive ratio of FIFO is at most
\(\left(m+1-\frac1m\right).
\)
\end{theorem}

The proof rests on a simple observation. Whenever at least $m$ jobs are alive,
each has a ready subjob, so FIFO keeps all $m$ processors busy. FIFO therefore
processes work at the maximum possible rate, and hence no slower than the
optimal schedule. We can use this observation to show that the total remaining work never grows beyond a threshold.

Let $Q(t)$ denote the total
remaining work of the jobs alive under FIFO (including subjobs not yet revealed) at time $t$.

\begin{lemma}
\label[lemma]{lem:fifo-remaining-work}
At every integer time $t$, we have
\(
Q(t)\le m^2\opt.
\)
\end{lemma}
\begin{proof}
If fewer than $m$ jobs are alive after the releases at time $t$, then
\Cref{prop:opt-bounds} bounds their total remaining work by
$(m-1)m\opt$ and the claim holds.

Suppose otherwise that at least $m$ jobs are alive at time $t$. Let $s\le t$ be the first timestep such that at least $m$ jobs are
alive at every time from $s$ through $t$. The minimality of
$s$ gives $Q(s-1)\le(m-1)m\opt$ (where $Q(-1)=0$ by convention).
By
\Cref{prop:opt-bounds}, the jobs released from timestep $s$ through
$t$ contain at most $m(t-s+\opt)$ subjobs. FIFO is fully busy during timesteps
$s,\ldots,t-1$, and therefore processes $m(t-s)$ subjobs. Hence
\[
Q(t)
\le
(m-1)m\opt+m(t-s+\opt)-m(t-s)
= m^2\opt. \qedhere
\]
\end{proof}

\begin{proof}[Proof of \Cref{thm:fifo-upper-bound}]
Fix a job $J$ released at time $r$, let $C$ be its completion time. Let
$G$ be the jobs released by time $r$, including $J$.

Call a timestep from $r$ through $C-1$ \emph{full} if FIFO processes subjobs
of jobs in $G$ on all $m$ processors, and \emph{partial} otherwise. Let $F$
and $P$ be the numbers of full and partial timesteps, respectively. It
suffices to bound $F$ and $P$.

In every partial timestep, FIFO processes every ready subjob of $J$. Every
longest path in the residual DAG begins at a ready subjob, so the longest
remaining path of $J$ decreases by at least one. Its initial length is at most
$\opt$ by \Cref{prop:opt-bounds}, and hence $P\le\opt$.

Every full timestep processes $m$ subjobs belonging to jobs in $G$, and every
partial timestep processes at least one subjob of $G$. Thus
\begin{equation}
\label{eq:fifo-full-partial-work}
mF+P\le Q(r)\le m^2\opt.
\end{equation}
The second inequality follows from \Cref{lem:fifo-remaining-work}. Therefore,
by~\eqref{eq:fifo-full-partial-work} and $P\le\opt$, we conclude
\[
C-r=F+P
\le m\opt+\left(1-\frac1m\right)P
\le \left(m+1-\frac1m\right)\opt. \qedhere
\]
\end{proof}

\subsection{Small-Frontier-First: An \texorpdfstring{$O(\opt)$}{O(opt)}-Competitive Algorithm}

The lower-bound construction above uses instances in which optimal maximum
flow time is large. Assuming we are given a priori value of $\opt$, we
now show a non-clairvoyant algorithm, Small-Frontier-First ($\SFF$), that is $O(\opt)$-competitive when
$\opt$ is bounded.

\begin{theorem}
\label{thm:opt-sensitive}
Given $\opt$, there is a deterministic non-clairvoyant algorithm $\SFF$ for online DAG
scheduling on $m$ identical processors with $F_{\max}(\SFF)\le 44\,\opt^{2}$.
\end{theorem}

\paragraph{Arrival batches.}
For each integer $j\ge 0$, the jobs released in the $2\opt$-length window $[2j\opt,\,2(j+1)\opt)$ form a
batch $\mathcal B_j$, whose DAG is the disjoint union of those jobs. A batch $\mathcal B_j$ becomes \emph{eligible} at its eligibility time $e_j=2(j+1)\opt$ and is \emph{admitted} at the
first round boundary at or after $e_j$.  A batch $\mathcal B_i$ is considered
\emph{older} than a batch $\mathcal B_j$ if their eligibility times give $e_i<e_j$.  

\begin{lemma}[Batch bounds]
\label[lemma]{lem:batch-bounds}
Each batch has at most $3m\opt$ subjobs and depth at most $\opt$.  For integers $u\le v$, at most $(v-u)/(2\opt)+1$ batches become eligible in
$(u,v]$, and together they hold at most $m(v-u+3\opt)$ total subjobs.
\end{lemma}
\begin{proof}

The \emph{depth} of a batch is the number of
vertices on the longest directed path of its DAG. Since each batch consists of jobs released in a $2\opt$-length window,
\Cref{prop:opt-bounds} gives depth at most $\opt$ and at most $3m\opt$ subjobs.

A batch $\mathcal B$ is eligible in $(u,v]$ only when its $2\opt$-length window ends within that interval, so
the windows of all such batches lie within $[u-2\opt,\,v)$ and hold at most
$m(v-u+3\opt)$ subjobs by \Cref{prop:opt-bounds}. Consecutive batch windows are $2\opt$ apart, so at most
$(v-u)/(2\opt)+1$ can become eligible in the interval.
\end{proof}

\paragraph{The algorithm.}
We construct blocks of consecutive
steps called \emph{rounds}. A round's endpoints are its \emph{boundaries}. $\SFF$ acts only at boundaries. At a boundary, an unfinished admitted batch is \emph{small} if it has at
most $m$ ready subjobs and \emph{large} otherwise. These labels use only the
subjobs ready at that boundary and are recomputed each round. $\SFF$ performs the following routine each round. 

\begin{algorithm}[H]
\renewcommand{\thealgorithm}{}
\captionsetup{labelformat=empty,labelsep=none}
\caption{\textbf{Small-Frontier-First Algorithm} -- single-round routine}
\label{alg:optupper}
\begin{algorithmic}[1]
\State Admit all eligible batches.
\State Let $P$ be the ready subjobs of the small batches.
\State Let $\mathcal{B}^*$ be the oldest large batch, if one exists.
\If{$P \neq \emptyset$}
    \State Process $P$ in a fixed order over $\lceil |P|/m\rceil$ steps.
    \State In the last step, fill free processors from $\mathcal{B}^*$, if possible.
\Else
    \If{$\mathcal{B}^*$ exists}
        \State Process $m$ ready subjobs of $\mathcal{B}^*$ for one step.
    \Else
        \State Leave all processors idle for one step.
    \EndIf
\EndIf
\end{algorithmic}
\end{algorithm}

 Rounds begin at time step $0$. Before any batch becomes eligible, every round is a single idle step. $\SFF$ holds $P$ and $\mathcal{B}^*$ fixed for the full round. Subjobs and batches that arrive during a round are deferred to the next round; consequently, an individual round has length $\max\{1,\lceil |P|/m\rceil\}$, a round with a large batch $\mathcal{B}^*$ must be fully busy, and any round has at most one idle step.

\paragraph{Progress in a round.}
The analysis builds on two observations. 
\begin{enumerate}
    \item A batch $\mathcal{B}$ can be small at the start of at most $\opt$ rounds. In each round, $\SFF$ processes all of its ready subjobs, so every longest path loses a vertex and a batch's depth drops by one whenever it is small. Thus at most $m\opt$ subjobs of a batch are processed over all rounds in which it is small.
    \item At most one batch may turn from large to small per round boundary, since only $\mathcal{B}^*$ is processed among large batches.
\end{enumerate}

\begin{lemma}[Round bounds]

\label[lemma]{lem:round-bounds}
Each round begins with at most $4\opt$ small batches and lasts at most $4\opt$
steps.
\end{lemma}

\begin{proof}
A round starting with $q$ small batches has length at most $\max\{1,q\}$, because
$|P|\le mq$. We seek to show $q\leq 4\opt$ by induction on rounds. 

Assume each earlier round had length at most $4\opt$. Then the $\opt$ rounds preceding a boundary time $t$ span at most $4\opt^2$
time steps, in which at most $2\opt+1$ batches become eligible (\Cref{lem:batch-bounds}).

Fix some batch $\mathcal{B}$ which is small at boundary time $t$. If $\mathcal{B}$ became eligible in this window, it is one of the aforementioned $2\opt +1$ batches. Otherwise, $\mathcal{B}$ was present at the beginning of all $\opt$ preceding rounds. Since $\mathcal{B}$ is small for at most $\opt$ rounds before finishing, $\mathcal{B}$ was large at the start of some previous round, and turned from large to small at one of the preceding boundaries. Since this can happen to at most one batch per boundary, over all $q$ small batches this gives $q\le(2\opt+1)+\opt\le 4\opt$.
\end{proof}

\begin{lemma}[Remaining work]
\label[lemma]{lem:remaining-work}
At every round boundary, the total remaining work of eligible batches, including
unrevealed subjobs, is at most $16m\opt^{2}$.
\end{lemma}

\begin{proof}
Let $W(t)$ be the total remaining work at boundary $t$, after all batches eligible by $t$ are admitted. With no large batch,
\Cref{lem:round-bounds} leaves at most $4\opt$
batches, each with at most $3m\opt$ work by \Cref{lem:batch-bounds}, so
$W(t)\le 12m\opt^{2}$.

Otherwise, let $u\le t$ be the latest
boundary with no large batch, which always exists since no batch is eligible at time $0$. We have $W(u)\le 12m\opt^{2}$ by the previous
argument.
The round at $u$ has at most one idle step, and all rounds after $u$ but preceding $t$ are
fully busy, so at least $m(t-u)-m$ work is processed in $[u,t)$. Any batch admitted after $u$ became eligible in $(u,t]$, so \Cref{lem:batch-bounds} gives that at most $m(t-u+3\opt)$ work is admitted. Remaining work rises only with admissions and falls only by processing, so the $m(t-u)$ terms cancel giving
\begin{align*}
  W(t) &\le W(u)+m(t-u+3\opt)-\big(m(t-u)-m\big) \\
       &= W(u)+3m\opt+m \le 16m\opt^{2}. \qedhere
\end{align*}
 \end{proof}

\begin{proof}[Proof of \Cref{thm:opt-sensitive}]
Fix a batch $\mathcal B$ admitted at boundary time $a$ with completion time $C$, and
set $T:=C-a$. The batches eligible at $a$ gain no further work and hold
at most $16m\opt^{2}$ work together (\Cref{lem:remaining-work}).

For every $t\in\{a,\ldots,C-1\}$, let $x_t$ be the total number of processors
$\SFF$ allocates to these batches during the interval $[t,t+1)$.  Since they gain
no work after $a$,
\[
\sum_{t=a}^{C-1}x_t\le 16m\opt^{2}.
\]
Each of the $m-x_t$ remaining processors is either idle or servicing a younger batch.  A
younger batch is served only while it is small for at most $m\opt$ work over its
lifetime, or as the large filler batch $\mathcal B^{*}$. Since $\mathcal B^{*}$ is the
oldest large batch, a younger $\mathcal B^{*}$ is only chosen when $\mathcal B$
is small. Furthermore, a step is idle only in rounds with no large batch, where
$\mathcal B$ is also small. At most $T/(2\opt)+1$ younger batches become
eligible in $(a,C]$ (\Cref{lem:batch-bounds}) and $\mathcal B$ is small in
at most $\opt$ rounds, so
\[
\sum_{t=a}^{C-1}(m-x_t)\le \frac{mT}{2}+3m\opt.
\]
It then follows that
\[
mT=\sum_{t=a}^{C-1}x_t+\sum_{t=a}^{C-1}(m-x_t)\le 16m\opt^{2}+\frac{mT}{2}+3m\opt,
\]
which gives $T\leq 32\opt^{2}+6\opt$. The same bound applied to every finite prefix of
unfinished rounds forces $\mathcal B$ to complete, and $\mathcal B$ completes within
$32\opt^{2}+6\opt$ steps after it is admitted.  A job of $\mathcal B$ is released at
most $2\opt$ steps before the batch becomes eligible and is admitted at most
$4\opt$ later (\Cref{lem:round-bounds}), then finishes with the rest of batch $\mathcal B$, thus its flow time is at most
\[
2\opt+4\opt+32\opt^{2}+6\opt\le 44\opt^{2}. \qedhere
\]
\end{proof}

\subsection{Combining FIFO and Small-Frontier-First}

We now show \Cref{thm:combined} by constructing a scheduler which combines FIFO and SFF and removes the assumption that we know $\opt$. 

By \Cref{thm:fifo-upper-bound}, FIFO is $O(m)$-competitive. By
\Cref{thm:opt-sensitive}, SFF is $O(\opt)$-competitive when supplied an upper
bound on $\opt$.
For any integer $B\ge\opt$, replacing $\opt$ by $B$ in the analysis of $\SFF$ gives a modified algorithm $\SFF_B$ with maximum flow at most $44B^2$.
Every job has depth at most $B$, and jobs released in an interval of length $\ell$ have at most $m(\ell+B)$ work, so all bounds in that analysis are valid.
Then running both under a doubling guess yields an algorithm that does not require knowledge of $\opt$ and
is $O(\min\{m,\opt\})$-competitive.

\paragraph{The combined algorithm.}
By \Cref{thm:fifo-upper-bound,thm:opt-sensitive} and the previous observation, we may fix a constant $c\ge44$
such that on every instance $I$, $F_{\max}(\SFF_B(I))\le cB^{2}$ whenever $\opt(I)\le B$ and
$F_{\max}(\FIFO(I))\le cm\cdot\opt(I)$. The algorithm runs in phases, each carrying an integer guess $B$ of $\opt$ and
running a fresh instance of $\SFF_B$, then $\FIFO$ once $B\ge m$.
At the start of a phase, each unfinished job is represented by the DAG obtained by deleting all completed subjobs. Future jobs retain their original DAGs. $\SFF_B$ treats the phase's start $s$ as its own time $0$, and presents each job
$J$ at global artificial release time $\hat r_J=r_J+cB^{2}$, or equivalently at local time $\hat r_J-s$.

\begin{algorithm}[H]
\renewcommand{\thealgorithm}{}
\captionsetup{labelformat=empty,labelsep=none}
\caption{\textbf{Combined Non-Clairvoyant Algorithm}}
\label{alg:combined}
\begin{algorithmic}[1]
\State Set $B\gets1$ and begin a phase.
\State For each unfinished or future job, set $\hat r_J\gets r_J+cB^2$.
\State Run $\SFF_B$ if $B<m$; otherwise, run $\FIFO$ permanently.
\State At each time, first record all completions.
\State An $\SFF$ phase fails if an unfinished job has age at least $2cB^2$.
\State On failure, set $B\gets2B$ and begin a new phase.
\end{algorithmic}
\end{algorithm}

\begin{proof}[Proof of \Cref{thm:combined}]
Assume $I$ is nonempty, so $\opt\ge1$.
All artificial releases in the initial phase occur after its start.
If a phase with guess $b$ fails at some time $s$, every unfinished job has age at most $2cb^2$ and under the new guess $B=2b$, each such job has artificial release
$r_J+cB^2\ge s-2cb^2+4cb^2=s+2cb^2>s$.
Future jobs have artificial releases after $s$, so each fresh $\SFF_B$ starts before any of its artificial releases.
Shift an optimal schedule of $I$ forward by $cB^{2}$ and restrict it to the
subjobs remaining at the start of the phase. This is a feasible schedule for
the phase instance, since restricting to the residual DAGs only removes
subjobs and precedence constraints, and it finishes each job within $\opt$ of
its artificial release. Thus, a phase's optimum is at most $\opt$.

Once $B\ge\opt$,
$F_{\max}(\SFF_B)\le cB^{2}$ guarantees each job finishes by age $2cB^{2}$, so an $\SFF$ phase never fails.
Then a failure requires $B<\opt$, and since each guess doubles a failed one, $B\le2\opt$.
Finitely many guesses lie below $m$,
so $\FIFO$ eventually runs or some $\SFF$ phase never fails and every job completes. 

Every guess satisfies $B\le2\opt$ and $B<2m$, so $B\le2\min\{m,\opt\}$ and
$B^{2}\le4\opt\min\{m,\opt\}$. A job completed under $\SFF$ has flow time at most
$2cB^{2}\le8c\opt\min\{m,\opt\}$. Under $\FIFO$, $m\le B\le2\opt$, so its flow time is at most
$cB^{2}+cm\opt\le6c\opt\min\{m,\opt\}$. Thus $F_{\max}\le8c\opt\min\{m,\opt\}$, and the ratio
is $O(\min\{m,\opt\})$.

\end{proof}

\section{FIFO Can be Bad Even When \texorpdfstring{$\opt=O(1)$}{opt = O(1)}}
\label{sec:fifo-short-jobs}

We next show that FIFO can have competitive ratio $\Omega(m)$ even when
$\opt$ is constant.

\begin{theorem}[FIFO lower bound with constant optimum]
\label{thm:fifo-constant-opt}
There exists an instance $I$ in which every job is an
out-tree and
\[
\opt(I)\le 6
\qquad\text{and}\qquad
F_{\max}(\FIFO(I))\ge 3(m-1).
\]
\end{theorem}

The high-level idea is to construct an instance where an optimal schedule
keeps all processors busy, while FIFO is repeatedly forced to leave a
processor idle.

\paragraph{The hard instance.}
Set $T:=3(m-1)$. For each $i\ge1$, release job $J_i$ at time
$r_i:=3i-2$. Each job has exactly $3m$ subjobs and is one of the following
types, illustrated in \Cref{fig:fifo-job-types}:
\begin{itemize}
\item A \emph{short job} consists of a root with $3m-1$ leaf children.
\item A \emph{long job} consists of a root whose unique child, called its
middle subjob, has $3m-2$ leaf children.
\end{itemize}

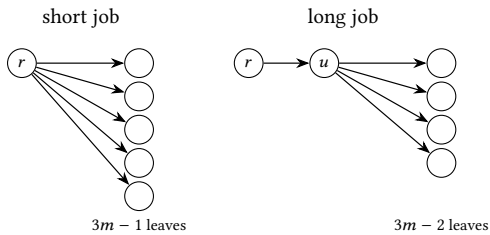
\begin{figure}[H]
\centering
\begin{tikzpicture}[
    >=Stealth,
    vertex/.style={draw,circle,minimum size=3.8mm,inner sep=0pt,font=\scriptsize},
    internal/.style={vertex},
    leaf/.style={vertex},
    dependency/.style={->,line width=0.45pt},
    font=\small
]
    \node[font=\small] at (0.78,1.25) {short job};
    \node[internal] (sr) at (0,0.66) {$r$};
    \node[leaf] (s1) at (1.55,0.66) {};
    \node[leaf] (s2) at (1.55,0.22) {};
    \node[leaf] (s3) at (1.55,-0.22) {};
    \node[leaf] (s4) at (1.55,-0.66) {};
    \node[leaf] (s5) at (1.55,-1.10) {};
    \foreach \v in {s1,s2,s3,s4,s5}
        \draw[dependency] (sr) -- (\v);
    \node[font=\scriptsize,anchor=base] at (1.55,-1.55) {$3m-1$ leaves};

    \node[font=\small] at (4.25,1.25) {long job};
    \node[internal] (lr) at (3.00,0.66) {$r$};
    \node[internal] (lm) at (4.00,0.66) {$u$};
    \node[leaf] (l1) at (5.55,0.66) {};
    \node[leaf] (l2) at (5.55,0.22) {};
    \node[leaf] (l3) at (5.55,-0.22) {};
    \node[leaf] (l4) at (5.55,-0.66) {};
    \draw[dependency] (lr) -- (lm);
    \foreach \v in {l1,l2,l3,l4}
        \draw[dependency] (lm) -- (\v);
    \node[font=\scriptsize,anchor=base] at (5.55,-1.55) {$3m-2$ leaves};
\end{tikzpicture}
\caption{The two job types. Labeled vertices are the root $r$ and, in a long
job, the middle subjob $u$; unlabeled vertices represent leaves.}
\label{fig:fifo-job-types}
\Description{A short job has a root joined directly to 3m minus 1 leaves. A
long job has a root joined to one middle subjob, which is joined to 3m minus 2
leaves.}
\end{figure}

The adversary initially reveals only the root of each job and assigns its type
adaptively. If FIFO schedules at least one root in timestep $t$, the adversary
adaptively assigns types as follows:
\begin{enumerate}[label=(\roman*)]
\item If the set also contains a middle subjob, declare every job whose root
is selected short.
\item Otherwise, declare every job whose root is selected long.
\end{enumerate}
If these roots are processed in timestep $t$, the corresponding children
become ready and are revealed at time $t+1$. Note that this adaptive choice is valid
as FIFO is non-clairvoyant.

\begin{claim}[Small optimum]
\label[claim]{clm:fifo-short-opt}
The instance $I$ generated above maintains that $\opt(I)\le6$.
\end{claim}
\begin{proof}
We describe a schedule with maximum flow time at most $6$. Consider the job $J_i$ released at time $r_i$. During timesteps
$r_i,r_i+1,r_i+2$, we reserve one processor for $J_i$ and process its root and
then one ready subjob in each of the next two timesteps. Its remaining
$3m-3=3(m-1)$ subjobs are then ready leaves. During timesteps
$r_i+3,r_i+4,r_i+5$, the other $m-1$ processors finish these leaves while the
reserved processor processes $J_{i+1}$, if it exists. Thus $J_i$ has flow
time at most six: its last subjob is processed by timestep $r_i+5$, so its
completion time is at most $r_i+6$.
\end{proof}

The adversary ensures that ready leaves belong to the oldest jobs, so FIFO
processes them first. Moreover, when leaves are ready, no middle subjob is
ready. Once FIFO has processed all leaves, it leaves a processor idle in the
next timestep.

\begin{claim}
\label[claim]{clm:fifo-short-invariant}
Suppose every job has flow time less than $T$. Then:
\begin{enumerate}[label=(\roman*)]
\item every job with a ready leaf is older than every other job with an unprocessed
root;
\item if a leaf is ready, then no middle subjob is ready; and
\item every block of $T+1$ consecutive timesteps contains a timestep in which
FIFO leaves a processor idle.
\end{enumerate}
\end{claim}
\begin{proof}
Since jobs are released three timesteps apart, among any $m$ alive jobs the
oldest has age at least $3(m-1)=T$ and hence flow time greater than $T$.
Thus fewer than $m$ jobs are alive at every time.

We prove (i) and (ii) by induction on time. Both trivially hold at time $1$.
Suppose (i) and (ii) both hold at time $t$. We show that they hold for time $t+1$.

If a leaf is ready at time $t$, then by the induction hypothesis (i), FIFO
processes leaves before roots;
by (ii), no middle subjob is ready. Thus every job whose root is processed
in timestep $t$ is declared long by the adversary's second rule
and reveals only its middle
subjob at time $t+1$. If some leaves remain after timestep $t$, no root was
processed, so at time $t+1$ these leaves still precede all unprocessed roots
and no middle subjob is ready. Otherwise, no leaf is ready at time $t+1$,
so (i) and (ii) hold trivially.

If no leaf is ready at time $t$, each alive job has one ready root or middle
subjob. Since fewer than $m$ jobs are alive, FIFO processes all of these
subjobs in timestep $t$. If a middle subjob runs in that timestep, every job
whose root runs is declared short, so all jobs alive at time $t$ reveal leaves
at time $t+1$. Otherwise, every job whose root runs is declared long, so only
middle subjobs will be revealed at time $t+1$.

Any root released at time $t+1$ belongs to a younger job. Hence (i) and (ii)
hold at time $t+1$ in both cases.

It remains to prove (iii). Consider a maximal run of timesteps with ready
leaves. At its start there are at most $3m\cdot(m-1)$ leaves, since fewer than
$m$ jobs are alive and each has $3m$ subjobs. By (i) and (ii), no new leaf is
revealed during the run. FIFO processes $m$ leaves in every timestep except
possibly the last, so the run lasts at most
$3m\cdot(m-1)/m=3(m-1)=T$ timesteps. Whenever no leaf is ready, each alive job
has exactly one ready root or middle subjob. Since fewer than $m$ jobs are
alive, FIFO processes all of them and leaves a processor idle. Thus every
block of $T+1$ consecutive timesteps contains a timestep with an idle processor.
\end{proof}

\begin{proof}[Proof of \Cref{thm:fifo-constant-opt}]
Suppose, for contradiction, that every job has flow time less than $T$.
At time $3i$, the first $i$ jobs have been released and contain $3mi$
subjobs, while only timesteps $1,\ldots,3i-1$ have been processed. Their
remaining work is therefore $m$ plus the number of idle processor slots in
these timesteps.

By part~(iii) of \Cref{clm:fifo-short-invariant}, this idle-slot count
increases in every block of $T+1$ timesteps. However, fewer than $m$ jobs are
alive, each with at most $3m$ remaining work, so the total remaining work is
always at most $3m\cdot(m-1)$, a contradiction.

Therefore $F_{\max}(\FIFO(I))\ge T=3(m-1)$, while
\Cref{clm:fifo-short-opt} gives $\opt(I)\le6$, completing the proof. 
\end{proof}

\section{Optimal Analysis of Batched Arrivals}
\label{sec:batched-bounds}
\label{sec:batched-lower-bound}
We now turn to the \emph{batched-arrival} setting, in which jobs are released
only at times $0,\opt,2\opt,\ldots$. Agrawal, Moseley, Newman, and
Pruhs~\cite{agrawal24} proved that FIFO is $O(\log(m\opt))$-competitive in
this setting. We give a somewhat simpler proof with a tighter
$O(\log\opt)$ analysis and prove a matching
$\Omega(\min\{m,\log\opt\})$ lower bound for every
non-clairvoyant algorithm. Together with FIFO's $O(m)$ guarantee from
\Cref{thm:fifo-upper-bound}, these results determine the optimal competitive ratio
up to constant factors, proving \Cref{thm:batched-optimal}.

\subsection{Upper Bound}
We now prove the upper bound in \Cref{thm:batched-optimal}. Fix a batched
instance $I$ and write $\opt:=\opt(I)$. We show that
\[
F_{\max}(\FIFO(I))
\le \bigl(\lceil\log_2\opt\rceil+2\bigr)\opt.
\]

Merge all jobs released at time $i\opt$
into a batch $\mathcal B_i$. By \Cref{prop:opt-bounds}, each batch contains
at most $m\opt$ subjobs. For $t\ge i\opt$, let $w_i(t)$ be the remaining work of
$\mathcal B_i$ under FIFO. We call a timestep \emph{prefix-idle for
$\mathcal B_i$} if FIFO processes fewer than $m$ subjobs from
$\mathcal B_0,\ldots,\mathcal B_i$, and let $z_i(t)$ be the number of such
timesteps before time $t$. Because FIFO prioritizes older batches, in every
prefix-idle timestep it processes every ready subjob of $\mathcal B_i$.

The following somewhat counterintuitive key lemma is due to Agrawal et
al.~\cite{agrawal24}.

\begin{lemma}
\label[lemma]{lem:batched-prefix-idle}
If $\mathcal B_i$ is unfinished at time $t$, then
\[
w_i(t)\le m\bigl(\opt-z_i(t)\bigr).
\]
In particular, $z_i(t)<\opt$.
\end{lemma}

\begin{proof}
In every prefix-idle timestep, FIFO processes every ready subjob of
$\mathcal B_i$. Thus the longest remaining path of every unfinished job in
$\mathcal B_i$ decreases by at least one. As every initial path has length at
most $\opt$ by \Cref{prop:opt-bounds}, we have $z_i(t)<\opt$.

Moreover, after $z_i(t)$ such timesteps, every remaining subjob has at least
$z_i(t)$ predecessors on some directed path. Hence an optimal schedule cannot
process any remaining subjob before time $i\opt+z_i(t)$ and must process all
remaining work by time $i\opt+\opt$. The claimed bound follows from the
fact that at most $m\bigl(\opt-z_i(t)\bigr)$ subjobs can be processed in these
$\opt-z_i(t)$ timesteps.
\end{proof}

For each $h$, let $U_h$ be FIFO's unfinished work immediately before the
batch at time $h\opt$ is released. We call $h\opt$ a \emph{strict record
time} if $U_h>U_j$ for every $j<h$.

\begin{lemma}
\label[lemma]{lem:batched-record}
Suppose $h\opt$ is a strict record time. If exactly $q$ of the
earlier batches $\mathcal B_0,\ldots,\mathcal B_{h-1}$ still contain
unfinished work immediately before $\mathcal B_h$ is released, then
\[
q\le \lceil\log_2\opt\rceil.
\]
\end{lemma}

\begin{proof}
Order the unfinished batches from oldest to newest. For the $r$th batch,
let $w_r$ be its remaining work and let $z_r$ be its number of prefix-idle
timesteps by time $h\opt$. We first show that, for every $r=1,\ldots,q$,
\begin{equation}
\label{eq:batched-geometric-progress}
z_r>\sum_{\ell=r+1}^{q}z_\ell.
\end{equation}

The $r$th batch has age at least $(q-r+1)\opt$. When it arrived, its prefix
contained less than $U_h+m\opt$ work because $h\opt$ is a strict record time. Every
timestep that is not prefix-idle removes $m$ units of this work. Hence
\[
\sum_{\ell=1}^{r}w_\ell
<U_h+m\opt-m\bigl((q-r+1)\opt-z_r\bigr).
\]
Since $U_h=\sum_{\ell=1}^{q}w_\ell$, rearranging and applying
\Cref{lem:batched-prefix-idle} gives
\[
z_r>
\sum_{\ell=r+1}^{q}\left(\opt-\frac{w_\ell}{m}\right)
\ge \sum_{\ell=r+1}^{q}z_\ell,
\]
which proves \eqref{eq:batched-geometric-progress}.

As the $z_r$ are integers, the recurrence above yields $z_r\ge2^{q-r}$.
As the oldest batch is unfinished, \Cref{lem:batched-prefix-idle}
gives $2^{q-1}\le z_1<\opt$, and the result follows.
\end{proof}

\begin{proof}[Proof of the upper bound in \Cref{thm:batched-optimal}]
The case $\opt=1$ is immediate, so assume $\opt\ge2$ and set
$K:=\lceil\log_2\opt\rceil$. At every strict record time, at most $K$
batches are unfinished, each with at most $m\opt$ subjobs. It follows that
$U_h\le Km\opt$ for every $h$.

Fix a batch $\mathcal B_i$. At its release, it and all earlier batches have
at most $(K+1)m\opt$ unfinished work. Until $\mathcal B_i$ completes, every
timestep either processes $m$ units of this prefix or is prefix-idle for
$\mathcal B_i$. There are at most $(K+1)\opt$ timesteps of the first type and,
by \Cref{lem:batched-prefix-idle}, at most $\opt$ of the second. Thus every job
in $\mathcal B_i$ has flow time at most $(K+2)\opt$, giving an
$O(\log\opt)$ competitive ratio. Combining this with FIFO's $O(m)$ bound
from \Cref{thm:fifo-upper-bound} proves the theorem.
\end{proof}

\subsection{Lower Bound}
We now prove the lower bound in \Cref{thm:batched-optimal}. This bound matches the
preceding FIFO guarantee and reveals a sharp contrast with general arrivals:
restricting releases to multiples of $\opt$ reduces the optimal dependence
on $\opt$ from linear to logarithmic.

More precisely, fix $m\ge2$ and an integer $1\le c\le\lfloor m/2\rfloor$, and set $L:=8c5^c$. We construct, for every deterministic non-clairvoyant algorithm $\ALG$, a fixed finite batched-arrival instance $I$ satisfying
$$
\opt(I)=L
\qquad\text{and}\qquad
F_{\max}(\ALG(I))\ge cL.
$$
Since $L=8c5^c$, we have $c=\Theta(\log\opt(I))$.

The extension to randomized algorithms follows by constructing a hard distribution through repeated copies of this instance and applying Yao's lemma, as in the proof of \Cref{thm:non-clairvoyant-lb}.

\paragraph{The Jobs.}
We use a construction similar to that in \Cref{sec:lower-bound}.  Let $n = cmL$. The instance $I$ is the following. For $i=0,\ldots,n-1$, one job $J_{i+1}$ is released at time $iL$. We call timesteps $iL,\ldots,(i+1)L-1$ \emph{epoch $i$}. Each job is a layered out-forest with key subjobs defined as before, except that each wide layer now contains $b:=2m-1$ subjobs. Each job has exactly $L/2$ wide layers and $L/2$ narrow layers, and hence exactly $mL$ subjobs.

The adversary chooses the order of the wide and narrow layers adaptively, subject to the following \emph{prefix condition}: every prefix of a job contains at least as many wide layers as narrow layers. This condition ensures that the optimal maximum flow time is exactly $L$, as shown next.
\begin{lemma}
\label[lemma]{lem:batched-opt} The constructed instance $I$ satisfies $\opt(I)=L$.
\end{lemma}
\begin{proof}
Each job has $mL$ subjobs, so $\opt(I)\ge L$. We now show that each job can be processed  entirely within its release epoch, thus showing $\opt(I) = L$. In the $j$th timestep of that
epoch, process the key of layer $j$ and $m-1$ non-key leaves.

Let $W_j$ and $N_j$ count the wide and narrow layers among the first $j$
layers. The preceding keys make all these layers available. Each wide layer
contributes $2(m-1)$ non-key leaves, and earlier timesteps have processed
$(m-1)(j-1)$ of them. Thus the number of ready non-key leaves is
\[
\begin{aligned}
2(m-1)W_j-(m-1)(j-1)
&=(m-1)(W_j-N_j+1)\\
&\ge m-1.
\end{aligned}
\]
The schedule is therefore valid and finishes every job in exactly $L$
timesteps, proving $\opt(I)=L$.
\end{proof}

\paragraph{The Adversary.}
As in \Cref{sec:lower-bound}, the adversary chooses the key subjob of each wide layer so that $\ALG$ finishes the entire layer before the next layer becomes ready. It remains to describe how the adversary adaptively chooses the order of wide and narrow layers.

\paragraph{Choosing the Layer Order.}
The adversary chooses the layer order so as to force $\ALG$ to leave processor slots idle. At the beginning of each epoch, call a job \emph{active} if it still has an unfinished or unused wide layer (a wide layer that is not yet exposed), and \emph{passive} otherwise. For each active job $J_i$, let $x_i$ count its wide subjobs processed in the current epoch.

After each timestep, let $D$ be the number of timesteps remaining in the epoch. An active job $J_i$ becomes passive once $x_i\ge bD$ or it has no wide layers remaining. While a job is active, its next layer is wide; once it becomes passive, its next layers are narrow for the remainder of the epoch.

If an idle processor slot occurs, the adversary has achieved its goal for the epoch. For the remainder of that epoch, whenever a new layer becomes ready, it is made wide whenever a wide layer remains, and narrow otherwise. At the beginning of the next epoch, the active/passive labels and counters are reset.

We next verify that the adversary's adaptive choices are always consistent with the required job structure.

\begin{claim}
\label[claim]{clm:batched-prefix}
At every time and for every job, the sequence of layers fixed so far satisfies the prefix condition: every prefix contains at least as many wide layers as narrow layers.
\end{claim}
\begin{proof}
We prove the claim by induction over epochs. Revealing a wide layer cannot violate the prefix condition.

If a job becomes passive because all wide layers are exhausted, the remaining layers are narrow and the prefix condition is immediate. Otherwise, suppose it becomes passive with $D$ timesteps remaining, so $x_i\ge bD$. Then $\ALG$ has processed subjobs from at least $D$ wide layers during the epoch. Since no narrow layers are added while the job is active, the current prefix has surplus at least $D$ wide layers. Thus at least $D$ narrow layers remain, and at most $D$ further layers can become ready before the epoch ends. Making these layers narrow therefore preserves the prefix condition.

After an idle processor slot occurs, placing all remaining wide layers immediately before all remaining narrow layers also preserves the condition.
\end{proof}

The next lemma shows that, as long as at most $c$ jobs are alive in an epoch, the threshold rule cannot keep all $m$ processors busy for an entire epoch.

\begin{lemma}
\label[lemma]{lem:batched-idle}
Every epoch that begins with at most $c$ alive jobs contains an idle processor-time slot.
\end{lemma}

\begin{proof}
Suppose, for contradiction, that every timestep in the epoch is full. Let $p\le c$ be the number of jobs that are active at the beginning of the epoch. For each such job $j$, let $D_j$ be the number of timesteps remaining when the job either becomes passive or completes. Note that $D_j$ is well defined: if the job does not complete earlier, then it becomes passive by the end of the epoch since the threshold $x_j\ge bD$ is automatic when $D=0$. Without loss of generality, index the jobs so $D_1\ge D_2\ge\cdots\ge D_p$, and set $D_0:=L$.

For $i=0,\ldots,p$, let $Y_i$ denote the number of wide subjobs processed during the final $D_i$ timesteps of the epoch. We derive a recurrence on $D_i$ by giving a lower and an upper bound on $Y_i$.

\paragraph{Lower Bound on $Y_i$.}
Since every processor is busy, exactly $m$ subjobs are processed in each timestep. At most one of these can be a narrow subjob from each alive job, and there are at most $c$ alive jobs. Hence at least $m-c$ wide subjobs are processed in every timestep, giving

\begin{align}\label{eq:y-lb}
Y_i\ge (m-c)D_i.
\end{align}

\paragraph{Upper Bound on $Y_i$.}
Consider a job $j>i$. One timestep before it becomes passive or completes, there are $D_j+1$ timesteps remaining and the job is still active. Hence $x_j<b(D_j+1)$. In the next timestep, at most $m<b$ additional wide subjobs of this job can be processed. Thus job $j$ processes fewer than $b(D_j+2)$ wide subjobs while active.

After a job becomes passive, no new wide layer is given to it, so it can process at most one additional wide layer, contributing at most $b$ wide subjobs. Summing over all jobs gives

\begin{align}\label{eq:y-ub}
Y_i\le b\left(c+\sum_{j=i+1}^{p}(D_j+2)\right).
\end{align}
Combining \eqref{eq:y-lb} and \eqref{eq:y-ub},
\begin{align}
(m-c)D_i
\le
b\left(c+\sum_{j=i+1}^{p}(D_j+2)\right).
\end{align}
Since $c \leq m/2$ and $b=2m-1$, we have $b/(m-c)<4$, and therefore

$$
D_i+2\le4c+2+4\sum_{j=i+1}^{p}(D_j+2).
$$
Solving this recurrence yields
$
D_i+2\le(4c+2)5^{p-i}.
$
Taking $i=0$ and using $p\le c$ gives
$$
D_0 +2 = L+2\le(4c+2)5^p\le(4c+2)5^c<L,
$$
which is a contradiction, and thus the epoch contains an idle processor time slot.
\end{proof}

\begin{proof}[Proof of the lower bound in \Cref{thm:batched-optimal}]
Suppose every job has flow time less than $cL$. Then at most $c$ jobs are alive at each epoch's start, so \Cref{lem:batched-idle} guarantees an idle processor-time slot in every epoch.

Each epoch releases $mL$ work but processes at most $mL-1$, leaving at least $n=cmL$ unfinished work after $n$ epochs. Yet every job released by time $(n-c)L$ must have completed by time $nL$, leaving at most $c-1$ jobs and $(c-1)mL$ unfinished work, a contradiction.

Thus some job has flow time at least $cL$; \Cref{lem:batched-opt} completes the proof.
\end{proof}

\section{Related Work}

We now review some closely related work.  

Scheduling parallel programs which can be modeled as DAGs has been studied extensively.  For a single job, if the structure is known, trees can be scheduled optimally using several algorithms~\cite{hu61,muntz70} while minimizing makespan for general DAGs is NP-hard when the number
of processors is part of the input~\cite{lenstra78}.  However, several theoretically good and practical schedulers have been designed such as greedy schedulers~\cite{graham69} and work-stealing schedulers~\cite{blumofe99}.  

In the setting where multiple jobs share a parallel machine, called the \emph{multiprogrammed setting}, there has been both theoretical and practical work~\cite{arora98,agrawal08,he08,agrawal19}.  Flow-time metrics have also been studied in a different non-clairvoyant precedence-constrained scheduling setting, though both its information model and objectives differ from the dynamic-multithreaded DAG model studied here~\cite{robert08,garg19}.  In the context of flow-time objectives, most of this work is in the context of \emph{speed augmentation} where an
$s$-speed $c$-competitive algorithm achieves a competitive ratio of $c$
when given processors $s$ times the speed of the optimal schedule.
Speed augmentation has long been a central tool for obtaining guarantees in non-clairvoyant scheduling~\cite{motwani94,edmonds00,kalyanasundaram00}.
An algorithm is scalable if, for every fixed $\epsilon>0$, it is
$(1+\epsilon)$-speed competitive with a ratio depending only on
$\epsilon$. Agrawal et al.~\cite{agrawalLLM16spaa} show that FIFO is
scalable for DAG jobs.   Scalable algorithms have also been designed for other flow-time objectives such as average flow time~\cite{agrawalLLM16soda} and throughput~\cite{agrawalLLM18}.  

A closely related model of parallel jobs is the \emph{speedup curves model}~\cite{pruhsRS10,moseleyZZ22}.  In this model, the job has phases of parallelism and in phase $i$, the job has $W_i$ work and completes $s_i(m)$ work given $m$ processors.  Generally, $s_i$ is a non-decreasing sublinear function to indicate that getting additional processors is never harmful and that there is decreasing marginal benefit for each additional processor given to a job.  The speedup curves model and the DAG model are fundamentally different; algorithmic and impossibility results do not appear to translate from one to another.  For instance, given $(1+\epsilon)$ speed, FIFO is $O(1)$-competitive for DAGs~\cite{agrawalLLM16spaa}, but no non-clairvoyant algorithm can be $o(\log m)$-competitive in the speedup curves model~\cite{pruhsRS10}.  Interestingly, the scheduler that achieves this result is not FIFO; FIFO is a very bad scheduler for the non-clairvoyant setting in the speedup curves model.

Without extra speed, all non-clairvoyant algorithms have an unbounded competitive ratio in the speedup curves model~\cite{moseleyZZ22}. The distinction between the models is important for interpreting these
lower bounds. In the speedup curves model, one source of difficulty is
\emph{inter-job} allocation: a non-clairvoyant scheduler does not know
a job's current speedup curve and may allocate many processors to a job
that cannot use them effectively. In the DAG model, the scheduler sees
each job's ready vertices and therefore knows its currently available
parallelism. However, this information does not reveal which ready
vertices will expose additional parallelism. Our lower-bound
construction exploits this uncertainty in \emph{intra-job} scheduling:
delaying these vertices creates periods of low parallelism, whereas a
clairvoyant scheduler can process them early and overlap work across jobs.

\section{Conclusions}

This paper essentially closes the problem of non-clairvoyant scheduling of DAG jobs for maximum flow time without speed augmentation; even for trees, the achievable competitive ratio is shown to be $\Theta(\min\{m, \opt\})$ where $m$ is the number of processors and $\opt$ is the maximum flow time achieved by the optimal offline algorithm.  We either need speed augmentation or we must know the structure of dependences when the job arrives in order to get better results.  For trees, one can achieve $\Theta(1)$-competitiveness~\cite{agrawal24} if the structure is known on arrival.
An important open question is whether structural information allows us to get a constant or even a logarithmic competitive ratio for general DAGs, or even for natural subclasses such as series-parallel DAGs.  More generally, what partial information—such as  advice about the critical path or coarse information about future parallelism—is sufficient to recover a good competitive ratio?

\section*{Acknowledgments}

The second and third authors thank Nikhil Bansal for several insightful
discussions about the problem. All four authors thank Nikhil Bansal and Kirk
Pruhs for bringing the two groups together after discovering their independent
work on the problem.

\section*{AI Disclosure}

OpenAI's ChatGPT 6 Astra was used extensively in the development of this
work. The authors first developed the lower-bound construction for batched
arrivals. Building on this construction, ChatGPT 6 generated and
refined the arguments leading to the tight upper and lower bounds for general
non-clairvoyant scheduling. AI tools were used in finding references and  preparing the
manuscript. The authors are responsible for all content in the
paper.

\clearpage
\balance
\bibliographystyle{ACM-Reference-Format}
\bibliography{Bibliograpghy}

\end{document}